\documentclass[11pt,letter]{article}
\usepackage[utf8]{inputenc}

\usepackage{amsmath}
\usepackage{amssymb}
\usepackage{amsfonts}
\usepackage{amscd}

\usepackage[round,authoryear]{natbib}

\usepackage{latexsym}
\usepackage[nointegrals]{wasysym}
\usepackage{graphicx}
\usepackage{xcolor}
\usepackage{hyperref}
\usepackage[all,cmtip]{xy}
\usepackage{bbold}

\usepackage{enumitem}
\usepackage{mathtools,bbm}
\usepackage{xpatch}
\usepackage{version,xspace}

\usepackage[margin=1in]{geometry}
\usepackage{hyperref}%
\hypersetup{colorlinks=true, linkcolor=blue, breaklinks=true, urlcolor=blue}
\usepackage{url,doi}

\usepackage[utf8]{inputenc}
\usepackage{amsfonts}
\usepackage{amsthm}
\usepackage{enumitem}
\usepackage{mathtools,bbm}
\usepackage{xpatch}

\usepackage{mathtools}
\usepackage{cuted}
\usepackage{mathtools,bbm}

\usepackage[caption=false,font=footnotesize]{subfig}
\graphicspath{{./../images/}{./../fig}}

\usepackage{enumitem}
\setlist{nosep}

\newcommand{\R}{\mathbb{R}}

\newcommand{\Nop}{\mathcal{C}}

\newcommand{\until}[1]{\{1,\dots, #1\}}

\newtheorem{theorem}{Theorem}[section]

\newtheorem{lemma}{Lemma}[section]

\newtheorem{remark}[theorem]{Remark} 

{
      \theoremstyle{plain}

}
\newtheorem{definition}{Definition}[section]

\DeclareSymbolFont{bbold}{U}{bbold}{m}{n}
\DeclareSymbolFontAlphabet{\mathbbold}{bbold}

\newcommand{\sign}{\operatorname{sign}}

\newcommand{\setdef}[2]{\{#1 \; | \; #2\}}
\newcommand{\fCO}{f_{\operatorname{I}}}

\usepackage{algpseudocode,algorithm,algorithmicx}

\usepackage{hyperref}%
\hypersetup{colorlinks=true, linkcolor=blue, breaklinks=true, urlcolor=blue}

\newcommand\oprocendsymbol{\hbox{$\square$}}
\newcommand\oprocend{\relax\ifmmode\else\unskip\hfill\fi\oprocendsymbol}

\DeclareSymbolFont{bbold}{U}{bbold}{m}{n}
\DeclareSymbolFontAlphabet{\mathbbold}{bbold}

\newcommand\blfootnote[1]{%
  \begingroup
  \renewcommand\thefootnote{}\footnote{#1}%
  \addtocounter{footnote}{-1}%
  \endgroup
}

\begin{document}

\title{Signed Network Formation Games and Clustering Balance}
\author{Pedro Cisneros-Velarde \and Francesco Bullo}

\maketitle

\begin{abstract}
  We propose a signed network formation game, in which pairs of individuals 
 strategically change the signs of the edges in a complete
  network. These individuals 
  are members of a social network who strategically
  reduce cognitive dissonances by changing their interpersonal
  appraisals. We characterize the best-response dynamics for this game and
  prove that its implementation can dynamically drive the network to a
  sociologically meaningful sign configuration called \emph{clustering
    balance}. In this configuration, agents in the social network form one
  or more clusters that have positive relationships among their members but
  negative relationships among members of other clusters. In the past,
  various researchers in the fields of psycho-sociology, political science,
  and physics have looked at models that explain the generation of up to
  two clusters. Our work contributes to these fields by proposing a simple
  model that generates a broader class of signed networks.
\end{abstract}

\section{Introduction}
\blfootnote{
  This work is supported by the U. S. Army Research Laboratory and the
  U. S. Army Research Office under grant number W911NF-15-1-0577. The
    views and conclusions contained in this document are those of the
    authors and should not be interpreted as representing the official
    policies, either expressed or implied, of the Army Research Laboratory
    or the U.S. Government.}
\blfootnote{Pedro Cisneros-Velarde (e-mail:
    pacisne@gmail.com) and Francesco Bullo (e-mail:
    bullo@engineering.ucsb.edu) are with the University of California, Santa Barbara.}
\subsection{Problem description}

Signed graphs or networks offer a natural representation of social systems
involving friendly and antagonistic relationships between their
members. These relationships can be interpreted as interpersonal
appraisals. In the social sciences literature, there have been several
specific sign configurations that have been deemed important in a social
network, including for example \emph{structural balance}, \emph{clustering
  balance}, \emph{ranked clusters of M-cliques model} and
\emph{transitivity model}~\citep{ECJ:89}. Any of these configurations
represents a particular notion of a \emph{social balance} structure. Whenever
the social network is undirected, i.e., the positive or negative
relationships among individuals are reciprocal, all notions of social
balance reduce to structural balance or clustering balance. Consider a
complete social network, i.e., one in which all individuals know each
other. Historically, structural balance is the first notion that has been
formulated in the seminal work by Heider~\citep{FH:44,FH:46}. It
characterizes the stable configurations of signs in a social network
according to four rules known as ``Heider's rules" that eradicate cognitive
dissonances among the appraisals of its members. As a result, in the
structure of the network, one or two antagonistic \emph{clusters} of
individuals can appear in the social network. Each member of a cluster has
positive ties with every other member of the same cluster and negative ties
to individuals outside it. After Heider's work, Davis introduced the
concept of \emph{clustering balance} in his seminal work~\citep{JD:67},
which admits any arbitrary number of clusters. Obviously, the largest
number of clusters a network can have is the number of agents in the social
network.


Heider's rules provide a distinction between these two notions of social
balance. Structural balance satisfies all four rules: 1)``the friend of my
friend is my friend", 2)``the enemy of my friend is my enemy", 3)``the
friend of my enemy is my enemy", 4)``the enemy of my enemy is my friend". As
a consequence, all \emph{triads} in the network (a triad is a cycle of
three nodes) are positive, i.e., the product of their edges are
positive. In contrast, clustering balance satisfies all rules but the
fourth one. As a consequence, the network admits positive triads and triads
with three negative edges.



Since the last decade, researchers have started to incorporate dynamic
models into the structural balance theory~\citep{XZ-DZ-FYW:15}, aiming to
explain how a network can change its sign configuration so that the network
eventually satisfies structural balance. For the interested reader, we
refer to the works
\citep{KK-PG-PG:05,SAM-JK-RDK-SHS:11,VAT-PVD-PDL:13,PJ-NEF-FB:13n,WM-PCV-GC-NEF-FB:17f,PCV-NEF-AVP-FB:19}
for models based on both discrete and continuous time (deterministic)
dynamical systems with real-valued appraisals, to the
works~\citep{TA-PLK-SR:05,TA-PLK-SR:06,RF-DV-SY-OM:07} for stochastic or
physics related updating models, and to the
works~\citep{AvdR:11,MM-MF-PJK-HRR-MAS:11} for models based on a game
theoretical updating of the appraisals. However, despite this growing body
of works, to the best of our knowledge, there has been little attention to
dynamic models addressing clustering balance in the literature. In fact,
only the work~\citep{AvdR:11} has addressed the question of clustering
balance, but, as we will see, in a different setting from ours. The
work~\citep{PJ-NEF-FB:13n} has addressed the problem of dynamic models for
other notions of social balance, but it does not particularly address the
case of clustering balance nor has a game theoretical formulation.  In this
paper, we propose to fill a gap in the literature of dynamic balance theory
by providing a new model of dynamic clustering balance under a game
theoretical framework that has a psycho-sociological motivation in its
formulation.

\subsection{Statement of contribution}

Our first contribution is the proposal of a novel game-theoretical model for dynamic clustering balance with sociologically meaningful interpretations and analyze the properties of its dynamics. An important characteristic of our game, is that its players are formed by couples of agents in the network, i.e., pair of agents take a joint action. Our model is shown to be explainable as best-response dynamics of myopic players, and the eradication of cognitive dissonances of the agents in the network. 
Our key theoretical result is to prove finite time convergence of our model under best-response dynamics to signed network structures that are related to a notion of Nash equilibrium in which clustering balance is possible to be achieved. In particular, convergence to a network satisfying clustering balance is guaranteed whenever the network has up to five agents, and, for larger networks, we present compelling numerical evidence that suggests that this also happens under generic initial conditions. 

We now place our dynamic model in the context of other game theoretical signed network formation games. To the best of our knowledge, a game-theoretical interpretation to dynamic social balance was first introduced in the seminal work by 
\cite{AvdR:11}. He proposed that if a single agent in the social network can alter multiple relationships at the same time according to some strategy and randomly act ``irrationally" (i.e., change her relationships against optimizing her utility function), then convergence towards structural and clustering balance can be achieved. One example of proposed strategy followed by any agent $i$ is a ``copying mechanism": $i$ will choose another agent $j$, and copy all the relationships $j$ has over other agents $k$ (e.g., if $j$ is friend with $k$, then $i$ will become a friend of $k$ too). However, $i$ may not alter its relationship with $j$. In contrast, our model only assumes that $i$ and $j$ jointly 
alter their mutual relationship, and that this change depends on the appraisals that any other agent $k$ has over both of them. Thus, our work does not use a copying mechanism. The issue with having an agent continuously changing multiple relationships simultaneously (and repetitively) is that it implies a greater cognitive burden on the agent herself. Moreover, as pointed out by van~de~Rijt, multiple updating of relationships requires multiple consents from the other agents, e.g., making a positive relationship required both parties to consent to be in peace, so that in reality, we can think the updating of relationships as if there are multiple pairs of agents changing their relationships simultaneously. Another difference from our model is that we assume rational or utility-maximizing players and thus we do not introduce stochasticity as van~de~Rijt's work does. As we will see later, numerical evidence suggests that our model predicts good convergence properties to clustering balance under generic initial conditions. Therefore, adding stochasticity to avoid the convergence to networks with no clustering balance is not necessary in our model. 

A second relevant work is the one by 
\cite{MM-MF-PJK-HRR-MAS:11}, which only deals with the case of structural balance. In their game setting, the network topology is not fixed: agents update their relationships by establishing positive or negative links, as well as deleting links. In contrast to traditional network formation games, they assume that agents can unilaterally create new links (no need for bilateral consensus). In the vein of van~de~Rijt's model, any selected agent needs to update her relationships with all other agents in the network simultaneously. This work introduces a utility function that each agent tries to maximize in order to reduce cognitive dissonances by enforcing the four Heider's rules, and from which our work takes inspiration with the crucial difference that we only enforce three Heider's rules. 
Although their model does not need a complete graph at the beginning, it is shown that after $O(n)$ steps of playing best-response dynamics (where $n$ is the number of agents in the network), the network becomes complete and  satisfy structural balance. 
Given this precedence, we decided in our model to simply assume the network is complete and fixed from the beginning. In fact, Malekzadeh et al. show that the incorporation of both creating and deleting edges  in their game made the computation of best-response policies NP-hard for any selected agent, something which is avoided in our model. 

Finally, we mention the work by 
\cite{TH:17}. This work proposes a signed network formation game in the context of complete networks (agents cannot delete links), motivated by observations of how groups of people display bullying behavior and the interplay between dominance and status in conflict networks. This work does an exhaustive analysis of networks that correspond to the game's concept of Nash equilibrium, and found that their signed structures can correspond either to networks with only positive relationships, or to a network that belongs to the notion of social balance of the transitivity model. In particular, they found that negative edges cannot be reciprocated among agents, which implies that clustering balance is not part of their analysis. They assume that agents optimize different utility functions and that any agent can alter all of its relationships. These utility functions are, as explained by Hiller, ``based on agent's incentives to bully and gang up on each other"; whereas our work is related to the eradication of cognitive dissonances. Finally, in contrast to our work and the ones mentioned above, the work by Hiller does not provide dynamics or time evolution rules that provide convergence results to a Nash equilibrium.




As mentioned before, we focus on the study of complete graphs, which has
the highest density of triads. Complete graphs are important to study as a
first understanding towards empirical data, besides being traditionally
important as a theoretical setting for the social sciences. Triads have
been playing an important role in social network analysis since many
decades, and empirical evidence over datasets of users of different social
media have remarked the abundance and protagonist role of triad structures
in the understanding of online social networks
\citep{JL-DH-JK:10,HH-JT-LL-JL-XU:15}, even though these real-world
networks are not complete.  Moreover, many empirical studies have indicated
the persistence and abundance of triads with all negative edges, a clear
violation of the classical balance in favor of the clustering one
\citep{JD:67,JL-DH-JK:10}. Therefore, the importance of a dynamic
clustering balance model is that it allows for a theoretical explanation of
such phenomena. We also remark that the study of undirected graphs,
as pointed out by~\citep{MM-MF-PJK-HRR-MAS:11}, is very important
since triads with reciprocated edges with the same sign are highly abundant
in the three datasets studied by~\citep{JL-DH-JK:10}, whereas triads with
reciprocated edges with different signs have a presence of much less than
$1\%$.

Finally, our work is also relevant for the economics and political sciences
literature since signed network formation games can model the generation of
networks of conflict among different parties. Examples of recent studies
are networks of military alliances~\citep{MDK-DR-MT-FZ:17}, and networks of
trade and international relationships among states~\citep{MOJ-SN:15}.

\section{Preliminary modeling}

Let $N=\until{n}$ be the set of nodes, assume $n\geq 3$, and let $G=(N,E)$
 be an undirected complete graph composed by $n$
nodes and an edge set $E$ such that $\{i,j\}\in E$ is the undirected edge
between nodes $i$ and $j$. Then, $G$ defines the structure or topology of a
social network composed by $n$ agents.
%
%
%
For any $\{i,j\}\in E$, let $x_{ij}\in\{-1,+1\}$ denote the negative or positive interpersonal \emph{appraisal} or relationship between $i$ and $j$. 
The \emph{appraisal network} is a signed undirected graph defined as 
$G_{X}=(N,\{x_{ij}\}_{i,j=1\;,i\neq j}^n)$ with $x_{ij}=x_{ji}$ for any $\{i,j\}\in E$. We assume that both $G$ and $G_X$ have no self-loops. We will use the terms \emph{network} and \emph{graph} interchangeably.

\begin{definition}[Balanced, unbalanced and neutral triads]
A triad is \emph{balanced} whenever it is a positive cycle, i.e., the appraisals associated with its edges have a positive product, and it is \emph{unbalanced} whenever it is a negative cycle. A triad is \emph{neutral} whenever it is unbalanced and all of its edges have associated negative appraisals.
\end{definition}

%
%
%
%

Given a set $A$, we denote the indicator function by $1_A(x)$, so that
$1_A(x)=1$ if $x\in A$, or $1_A(x)=0$ if $x\notin A$. We omit the argument
in the indicator function whenever it is clear form the context. In what
follows, let $\sign(u)\in\{-1,0,+1\}$ denote the sign of $u\in\R$ (with
$\sign(0)=0$).

\begin{definition}[Cognitive dissonance function]
The \emph{cognitive dissonance function} defined on the appraisal network $G_X$ is given by
\begin{equation}
\label{pot_triad}
\Nop(G_X) = \sum_{\{i,j,k\}\in\mathcal{T}}
1_{\{x_{ij}=-x_{jk}x_{ki}\text{ AND }(x_{ij}= +1 \text{ OR }x_{jk}= +1 \text{ OR }x_{ki}= +1)\}},
\end{equation}
where $\mathcal{T}$ is the set of all triples of nodes that form a triad in the network, i.e., $|\mathcal{T}|={{n}\choose{3}}$.
\end{definition}

The cognitive dissonance function is equal to the total number of unbalanced triads in the network that are not neutral.
%

\begin{definition}[Clustering balance~\citep{JD:67}]
\label{balance2}
Consider an appraisal network $G_X$. We say $G_X$ has \emph{clustering balance} if there exists a partition of the $n$ agents into $k$ sets called \emph{clusters} or \emph{factions}, with $k\in\until{n}$, such that all appraisals between members of the same faction are positive and all appraisals between
  members of different factions are negative. Whenever $k\in\{1,2\}$, the
  network also satisfies \emph{structural balance}.
\label{tresb}
\end{definition}
The following lemma follows from a characterization given in~\citep{JD:67} and from the previous definition of the cognitive dissonance function.
\begin{lemma}
\label{balance2_l}
Consider an appraisal network $G_X$. The following statements are equivalent:
\begin{enumerate}[label=(\roman*)]
\item $G_{X}$ has clustering balance; \label{unob}
\item The number of balanced triads plus the number of neutral triads is
  equal to $|\mathcal{T}|$, i.e., $\Nop(G_X)$ is equal to zero.
  \label{dosb}
\end{enumerate}
Moreover, $G_X$ has structural balance when the number of balanced triads is $|\mathcal{T}|$.
\end{lemma}

\begin{remark}[Cognitive dissonances]
\label{remark1}
Given an appraisal network $G_X$ that has clustering balance,~\cite{JD:67} 
shows that this is equivalent to enforce the satisfaction of the first, second and third Heider's rules. The fourth rule is not enforced: this assumption translates into the acceptance of triads that have three negative appraisals.
%
Violations of any of these three rules generate incoherence in the agent's cognitive system and in her social environment, also known as \emph{cognitive dissonances}, that she strives to resolve~\citep{JW:15,LF:1957}. For example, an individual whose relationships violate the second Heider's rule 
might ask herself: ``how is it possible that I like the enemy of my friends?" In order to resolve such dissonances, 
individuals seek to transform the interpersonal appraisals in the triad so that it becomes a balanced triad or a 
neutral one. Then, we can interpret the cognitive dissonance function as a measure of the 
total level of eradication of cognitive dissonances in the appraisal network: the more
balanced and neutral triads, the less cognitive dissonances we expect among members of
the social network. 
%
%
%
\end{remark}

\section{Game theoretical formulation and static analysis}

\begin{definition}[Signed network formation game]
\label{def-game}
We define the following \emph{signed network formation game}. 
\begin{itemize}
\item The normal form of the game
is composed of the set of players $E$, where any player $\{i,j\}\in E$ has 
\begin{enumerate}
\item an action space $\mathcal{A}_{ij}=\{-1,+1\}$ such that, for any action $a_{ij}\in\mathcal{A}$ performed by the player: $a_{ij}=-1$ means setting the appraisal $x_{ij}=-1$ and $a_{ij}=+1$ means setting the appraisal $x_{ij}=+1$,
\item a payoff or utility function which is defined, given the action profile $a=(a_{ij})$, by
\begin{equation}
\label{payoff_1}
u_{ij}(a) = \Delta_{ij}^{b} - \Delta_{ij}^{u}-\lambda_{ij} x_{ij}1_{\{\Delta_{ij}^u>0\}},
\end{equation}
where $\Delta_{ij}^{b}$ and $\Delta_{ij}^{u}$ are the number of balanced and unbalanced triads the agents $i$ and $j$ are part of, and $\lambda_{ij}$ is the number of other agents that are enemies of both $i$ and $j$, i.e., that have negative relationships with both $i$ and $j$.
\end{enumerate}
%
\item In the game's dynamic setting, consider that any time step $t$ belongs to the countable infinite set $\{0,1,\cdots\}$, and:
\begin{enumerate}[label=(\roman*)]
\item At any time $t$, a player $\{i,j\}\in E$ is randomly and independently picked for its updating, with all possible pair of nodes having a positive probability of being picked according to some fixed time-invariant 
distribution. \label{def111}
\item The player $\{i,j\}$ selected at time $t$ chooses an action $a_{ij}(t+1)$ to maximize its current utility $u_{ij}(t)\equiv u_{ij}(a(t))$ in the next time step. The player does not change its action if such change does not strictly improve its utility.
\end{enumerate}
%
\end{itemize}
\end{definition}

\begin{remark}\label{remarkin}
 In the formal definition of the game, the player of the game is a pair of agents $\{i,j\}\in E$. Then, we will use the terms \emph{player} and \emph{pair} or \emph{couple} of agents indistinctly, with the term \emph{agent} continuing to refer to any member of the social network. Having pairs of agents taking a ``joint action" as a single player is motivated from the fact that many social interactions and relationships (e.g., friendships and enmity), usually require some level of consent and/or are naturally bilateral.
\end{remark}

In the definition of our game, it is assumed that players are myopic because they only want to maximize their current utility. 
Intuitively, when any couple of agents update their appraisals, they ignore how their actions can affect future decisions of other couples in the network, i.e., they ignore the global effect of their actions in the future evolution of the appraisal network. We point out that the assumption of selecting only one pair of agents 
per time step in our game's dynamic setting is a common assumption found in the literature of dynamic network formation games (e.g., see~\citep{MJ-AW:02}). 
The payoff~\eqref{payoff_1} associated with the edge $\{i,j\}$ 
can be interpreted as the total cognitive dissonance load for $i$ and $j$ 
resulting from the relationships with the other agents in the network. The maximization of this utility results in the reduction of cognitive dissonances (see Remark~\ref{remark1}). We remark that, in the field of social-psychology, the eradication of cognitive dissonances has been considered a fundamental model for driving human decision processes~\citep{LF:1957}.
Notice that for the computation of $u_{ij}$, both agents $i$ and $j$ effectively count the number of balanced and unbalanced triads they belong to, without considering the triads that contain an agent who has negative relationships with both of them 
(i.e., an agent who is a common enemy), and compute their difference. If $u_{ij}\geq 0$, then both agents see that their current interpersonal appraisal or relationship is appropriate to relieve most of the cognitive dissonances since the number of balanced triads is enough to counteract the number of unbalanced ones. On the other hand, if $u_{ij} < 0$, then both agents have a cognitive discomfort due to the majority of unbalanced triads they are part of. For example, consider a game with only one triad formed by the agents $i$, $j$ and $k$ such that $x_{ij}=x_{ik}=+1$ and $x_{kj}=-1$. In the perspective of $i$, this triad violates the Heider's rule ``the enemy of my friend is my enemy", since the enemy of $k$, which is $j$, has a positive relationship with $i$. Similarly, in the perspective of $j$, this triad violates the rule ``the friend of my enemy is my enemy". According to our game, $i$ and $j$ can switch their interpersonal appraisal to being negative and thus satisfy the appropriate Heider's rules and reduce cognitive dissonances for both of them.

We consider that players can only play pure strategies, then, we can also refer to any action profile as a \emph{pure strategy profile} indistinctly. As a solution concept for our proposed signed network formation game, we adopt a notion of pure Nash equilibrium network (see, for example, the work~\citep{CAA-IK:09}).

\begin{definition}[Nash equilibrium network]
\label{def_Nash}
A \emph{Nash equilibrium} is a pure strategy profile $a^*=(a^*_{ij})$ 
such that, for any player $\{i,j\}\in E$, we have 
$u_{ij}(a^*)\geq u_{ij}(a_{ij},a^*_{-ij})$\footnote{Given the action profile $a=(a_{ij})$, we use the notation $(a'_{ij},a_{-ij})$ to denote a new action profile resulting from only changing the action of player $\{i,j\}$ by $a'_{ij}$ in $a$.} 
for any $a_{ij}\in \mathcal{A}_{ij}$. An appraisal network $G_X$ is a \emph{Nash equilibrium network} if there exists a Nash equilibrium 
$a^*$ that induces the formation of 
$G_X$. 
\end{definition}

Intuitively, a Nash equilibrium network is a network such that no 
pair of agents have an incentive to unilaterally change the sign of their interpersonal appraisals. 

%
%
\begin{definition}[Efficient networks]
An~\emph{efficient network} is any appraisal network $G_X$ induced by some action profile $a$ such that the \emph{social welfare} quantity $v(G_X)=\sum_{\{i,j\}\in E}u_{ij}(a)$ is maximized. 
\end{definition}

The study of efficient networks has a long-standing history in the literature on network formation games~(e.g.,~\citep{MJ-AW:02,MOJ:05}). In our case, we are interested in knowing what are the distribution of interpersonal appraisals that can lead to the maximization of the social welfare.

\begin{lemma}[Characterization of Nash equilibrium networks]
\label{lem1a} 
The set of Nash equilibrium networks is the set of appraisal networks such that $u_{ij}\geq 0$ for any $\{i,j\}\in E$.
\end{lemma}
\begin{proof}
Consider an appraisal network induced by some underlying pure strategy profile $a^*=(a_{ij}^*)$ such that $u_{ij}(a^*)\geq 0$ for any $\{i,j\}\in E$. Now, let any $\{m,n\}\in E$ choose 
$a_{mn}\in A_{mn}$ with $a_{mn}\neq a_{mn}^*$ and define the new pure strategy $a=(a_{mn},a^*_{-mn})$. Then, $u_{mn}(a^*)\geq 0$ and $u_{mn}(a)\leq 0$, so that $u_{mn}(a^*)\geq u_{mn}(a)$. Then, by Definition~\ref{def_Nash} this appraisal network is a Nash equilibrium network.

To prove the converse, consider a Nash equilibrium network. Since no pair of agents wants to deviate their pure strategy, it follows that they have some strategy profile $a^*=(a^*)$ such that $u_{ij}(a^*)\geq u_{ij}(a)$ for any $\{i,j\}\in E$ and $a=(a_{ij},a^*_{-ij})$ with $a_{ij}\in\mathcal{A}_{ij}$. We claim that $u_{ij}(a^*)\geq 0$ for any $\{i,j\}\in E$. Otherwise, assume that there exists some $\{m,n\}\in E$ such that $u_{mn}(a^*)<0$, then we observe that $a=(-a^*_{mn},a^*_{-mn})$ gives $u_{mn}(a)>u_{ij}(a^*)$, thus leading to a contradiction.
%
%
%
%
 
\end{proof}

%

\begin{theorem}[Nash equilibrium networks and clustering balance]
\label{lem_1co}
All appraisal networks that have clustering balance are Nash equilibrium networks. Moreover, for $n\leq 5$, the set of Nash equilibrium networks is the set of networks that have clustering balance; and, for $n > 5$, there exist Nash equilibrium networks that do not have clustering balance.
\end{theorem}
\begin{proof}
We first prove the first statement of the lemma. Let~$\Delta^n_{ij}$ be the number of neutral triads in which agents $i$ and $j$ are part of. From Lemma~\ref{balance2_l}, an appraisal network $G_X$ that has clustering balance is such that, for any $\{i,j\}\in E$, $\Delta_{ij}^u=\Delta_{ij}^n$. Now, if $x_{ij}=-1$ then $\lambda_{ij}=\Delta_{ij}^n\geq 0$; and, additionally, if $\{i,j\}$ belongs only to triads with all negative edges, then $\Delta_{ij}^b=0$, otherwise, $\Delta_{ij}^b>0$. From this, it follows that $u_{ij}=\Delta_{ij}^b\geq 0$ for any $\{i,j\}\in E$. Now, if $x_{ij}=+1$ then $\Delta^u_{ij}=\Delta^n_{ij}=0$ and clearly $\Delta_{ij}^b>0$, so that $u_{ij}=\Delta^b_{ij}>0$. From these two cases, it follows from 
Lemma~\ref{lem1a} that $G_X$ is a Nash equilibrium network. This finishes the proof for the first statement of the lemma.
%
%
%
%
%
%

We proceed to prove that, for $n\leq 5$, the set of Nash equilibrium networks is the set of networks that have clustering balance. The case $n=3$ is immediate. Consider $n=4$. Let us arbitrarily label the vertices by elements of the set $\{1,2,3,4\}$. By contradiction, assume there exists at least one unbalanced triad which is not neutral, i.e., a non-neutral unbalanced triad. We will attempt to construct a Nash equilibrium network containing at least this non-neutral unbalanced triad by continuously verifying if any possible constructed appraisal network satisfies the characterization given in Lemma~\ref{lem1a}. Without loss of generality, let $\{1,2,3\}$ be a non-neutral unbalanced triad and let $x_{12}=-1$, so that $x_{23}=x_{31}=+1$. Then, it follows that
\begin{align*}
(-1)x_{24}x_{14}&=\rho_1,\\
(+1)x_{14}x_{34}&=\rho_2,\\
(+1)x_{24}x_{34}&=\rho_3;
\end{align*}
with the only possible consistent cases 
$$(\rho_1,\rho_2,\rho_3)\in\{(-1,-1,-1),(-1,+1,+1),(+1,-1,+1),(+1,+1,-1)\}.$$
Assume $(\rho_1,\rho_2,\rho_3)=(-1,-1,-1)$. For any $x_{24}\in\{-1,+1\}$, it follows that $u_{ij}<0$ for any $\{i,j\}\in E$, from which we conclude that the constructed network cannot be a Nash equilibrium one. 
Assume $(\rho_1,\rho_2,\rho_3)=(-1,+1,+1)$. For any $x_{24}\in\{-1,+1\}$, it follows that $u_{12}<0$ and the network cannot be a Nash equilibrium one. Assume $(\rho_1,\rho_2,\rho_3)=(+1,-1,+1)$. For any $x_{24}\in\{-1,+1\}$, it follows that $u_{13}<0$ and the network cannot be a Nash equilibrium one. Finally, assume $(\rho_1,\rho_2,\rho_3)=(+1,+1,-1)$. For any $x_{24}\in\{-1,+1\}$, it follows that $u_{23}<0$ and the network cannot be a Nash equilibrium one. 
Therefore, we conclude that it is not possible to construct a Nash equilibrium network that contains at least one non-neutral unbalanced triad. We conclude by contradiction, that all Nash equilibrium networks have clustering balance. 

From the previous analysis we make the following observation: if a complete graph of four vertices contains an unbalanced and a balanced triad, then it must contain one more additional balanced triad. With this observation, we can also prove the nonexistence of Nash equilibrium networks that do not satisfy clustering balance for the case $n=5$ after some careful algebraic and combinatorial analysis. See Figure~\ref{plotin_1} for all the possible 
cases of 
appraisal networks that have clustering balance.
  
%

Finally, assume $n=6$. We now construct a Nash equilibrium network that does not have clustering balance.  Assume the appraisal network has all its initial appraisals positive. Then, set the following: $x_{12}=x_{23}=x_{35}=x_{51}=-1$. 
It is easy to check that this network is a Nash equilibrium network that does not have clustering balance (e.g., observe that $u_{12}>0$ with the triad formed by nodes $\{1,2,4\}$ being unbalanced with two positive appraisals and a negative one). See Figure~\ref{plotin_1} for an illustration of this appraisal network.  Now, for any $n>6$, set again all appraisals in the appraisal network to be negative and then choose six of its nodes. For these nodes, construct a subgraph exactly as the example we just did for $n=6$, and we immediately see that this appraisal network is a Nash equilibrium network.
 
\end{proof}

 \begin{figure}[t]
   \centering
   \subfloat{\includegraphics[width=0.90\linewidth]{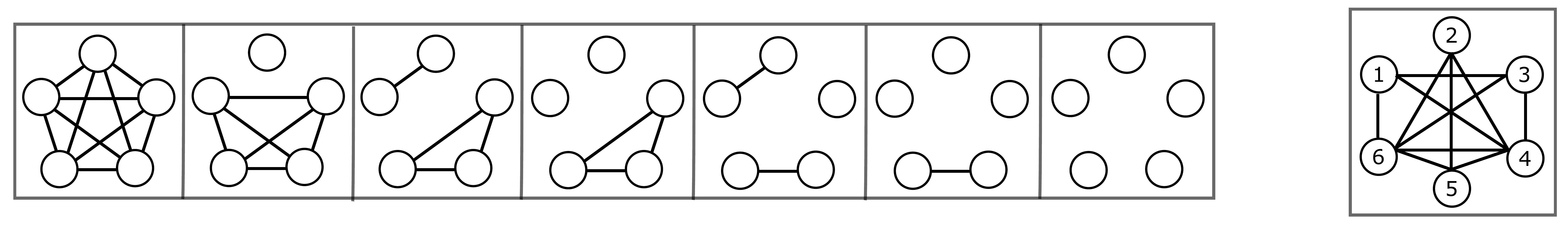}} 
 \caption{The left image shows all seven possible configurations for the (complete) appraisal network with five agents that has clustering balance. The right image shows an example of a Nash equilibrium network for six agents that does not have clustering balance. Whenever an edge is present, it corresponds to a positive appraisal; otherwise, we assume there is a negative appraisal.
 }
    \label{plotin_1}
 \end{figure}

\begin{lemma}[Efficient networks]
\label{cor_1co}
All efficient networks are Nash equilibrium networks and they have structural balance.
\end{lemma}
\begin{proof}
For any $\{i,j\}\in E$, we observe that this edge belongs to $n-2$ triads in the network, and so $u_{ij}\leq n-2$ under any action profile. Then, $u_{ij}=n-2$ if and only if $\Delta^b_{ij}=n-2$, $\Delta^u_{ij}=0$. 
Since this must hold for all edges in the appraisal network in order to maximize the social welfare $v(G_X)$ 
(which reaches the value ${n\choose{2}}(n-2)$), it follows that such network is a Nash equilibrium network. Finally, since all triads in $G_X$ are positive, then, from Lemma~\ref{balance2_l}, it immediately follows that this network has structural balance.
\end{proof}

An interesting result from Lemma~\ref{cor_1co} is that a network in which all agents have positive relationships with each other, i.e., all interpersonal appraisals are positive, is efficient and socially optimal because it maximizes the social welfare. Recall that the cognitive dissonance function is also minimized. As an additional observation, if, on the other hand, we only have 
negative appraisals, then the appraisal network has clustering balance and both the social welfare and cognitive dissonance function take value zero.

\section{Dynamic analysis}
  \begin{definition}[Influence dynamics]
    Consider $G_X(0)=(N,\{x_{ij}(0)\}_{i,j=1\;,i\neq j}^n)$ an initial appraisal network 
and that at each time $t\in\{0,1,\dots\}$, only one edge in $E$ is selected. For any selected $\{i,j\}\in E$, the \emph{influence dynamics} is defined by
  \begin{equation}
    \label{eq_CODI}
    x_{ij}(t+1) = \begin{cases}        
      \sign(\fCO(\{i,j\},G_X(t))), \qquad & \text{if } \fCO(\{i,j\},G_X(t))\neq 0,\\
      x_{ij}(t), & \text{otherwise},\\
    \end{cases}
  \end{equation}
  where $\fCO(\{i,j\},G_X(t)) = \sign{\left(\sum_{\substack{k=1\\k\neq
        i,j}}^n{x_{ik}(t)x_{kj}(t)1_{\{x_{ik}(t)= +1 \text{ OR }x_{jk}(t)=
        +1\}}}\right)}$. The appraisals associated to the rest of the unselected edges are left unchanged. 
\end{definition}

Notice that the dynamics \eqref{eq_CODI} is well defined because $x_{ij}(t)=\pm 1$ for any $\{i,j\}\in E$ and any $t>0$. 

The term ``influence dynamics" comes from a sociological interpretation of the fact that the updating of $x_{ij}$ depends on $x_{ik}x_{kj}$ with $k\neq i,j$. For example, in the perspective of agent $i$, the updating of the interpersonal appraisal or relationship she has with $j$ (i.e., the term $x_{ij}$) considers the influence that $k$ has over $i$ subject to what $k$ thinks of $j$ (i.e., the product term $x_{ik}x_{kj}$). Notice that whenever both $i$ and $j$ are enemies of $k$, the influences that $k$ has over $i$ and $j$ are not considered (i.e., $k$ cannot be trusted, since she is a common enemy). Intuitively, this implies that the fourth Heider's rule  ``the enemy of my enemy is my friend" is not enforced by the agents if they change their interpersonal appraisals according to the influence dynamics.
%
%


Now, from the assumptions of utility-maximizing and myopic players in Definition~\ref{def-game}, couples should play the game according to some policy such that they choose actions in the current time step that will maximize the payoffs associated to their interpersonal appraisals in the next time step, 
i.e., they should play~\emph{best-response dynamics}. Every pair of agents want to optimally change their current relationship in order to reduce cognitive dissonance. Recall that whenever $u_{ij} =0$, $\{i,j\}$ has no incentive to flip the sign of its appraisal since this will not increase its utility. Then, the following theorem provides a characterization of best-response dynamics taken by the agents.
%

\begin{theorem}[Best-response dynamics]\label{thm:best-response}
For any $\{i,j\}\in E$ chosen at any time step $t\geq 0$, let its action by updated by $a_{ij}(t+1) = x_{ij}(t+1)$, where the appraisal $x_{ij}(t+1)$ is updated according to the influence dynamics~\eqref{eq_CODI}. Then, the strategy $a(t)$, $t>0$, is a \emph{best-response dynamics} for the proposed signed network formation game. 
\end{theorem} 
\begin{proof}
Let us first note that, from the fact that players are utility-maximizing 
and the definition of the game, for any chosen $\{i,j\}\in E$, the following is the best-response by player $\{i,j\}$ 
in order to maximize its utility~\eqref{payoff_1}:
\begin{equation}
\label{eq_brd}
  a_{ij}(t+1) = \begin{cases}
    \underset{a_{ij}(t)\in\{+1,-1\}}{\text{arg max}}u_{ij}(t), \qquad & \text{if } u_{ij}(t)\neq 0,\\
       a_{ij}(t), & \text{otherwise},\\
  \end{cases}
\end{equation}
Now, assume some fixed action profile $a(t)$. Let $P_{ij}(t)=\setdef{k\in N}{x_{ik}(t)x_{kj}(t)>0\text{ with }x_{ik}(t)=x_{kj}(t)= +1}$ and $N_{ij}(t)=\setdef{k\in N}{x_{ik}(t)x_{kj}(t)<0}$. 
Then, from the fact that balanced triads are ones such that the product of two of its appraisals are equal in sign to the remaining appraisal, it follows the best-response for player $\{i,j\}$ in~\eqref{eq_brd} satisfies the following rule: 
\begin{itemize}
\item If $|P_{ij}(t)|>|N_{ij}(t)|$, then choosing $a_{ij}(t+1)=+1$ leads to $u_{ij}(t+1)>0$ (otherwise, we would have $u_{ij}(t+1)<0$).
\item If $|P_{ij}(t)|<|N_{ij}(t)|$, then choosing $a_{ij}(t+1)=-1$ leads to $u_{ij}(t+1)>0$ (otherwise, we would have $u_{ij}(t+1)<0$).
\item If $|P_{ij}(t)|=|N_{ij}(t)|$, then the player keeps the same action, i.e., $a_{ij}(t+1)=a_{ij}(t)$.
\end{itemize}
Then, it follows from the mathematical expression of the influence dynamics in~\eqref{eq_CODI} that $\fCO(\{i,j\},G_X(t)) = \sign(|P_{ij}(t)|-|N_{ij}(t)|)$ and then by setting 
$a_{ij}(t+1)=x_{ij}(t+1)$, $t\geq 0$, the player $\{i,j\}\in E$ is implementing best-response dynamics.
%
 
\end{proof}

Theorem~\ref{thm:best-response} states that, if pairs of agents update their
appraisals according to the influence dynamics, they are indeed playing 
best-response dynamics 
for the game. Then, it only remains to study the properties of convergence to Nash equilibrium networks of the trajectories of the dynamics~\eqref{eq_CODI} under the definition of our signed network formation game (which implies that these trajectories will now have a stochastic nature).
%

The following theorem establishes the convergence of the best-response dynamics played in the game, and we remark that any probability statement is respect to the induced measure by the underlying stochastic process that selects the pair of agents in the signed network formation game (see Definition~\ref{def-game}).

%
\begin{theorem}[Convergence of the influence dynamics]
\label{thm_main}
Consider the initial appraisal network $G_X(0)=(N,\{x_{ij}(0)\}_{i,j=1\;,i\neq j}^n)$. 
Then, under the influence dynamics~\eqref{eq_CODI} with the players being randomly selected as in Definition~\ref{def-game}, 
$G_X(t)$ converges with probability one to a Nash equilibrium network in finite time. Moreover, $G_X(t)$ converges to clustering balance with probability one for $n\leq 5$. 
\end{theorem}

\begin{proof} 
We first claim that the influence dynamics in our signed network formation game defines a 
time-homogeneous finite-state Markov chain $\mathcal{M}$ with state space $\mathcal{B}$ of all appraisal networks $G_{X}$, with the state $G_{X}(t)=(N,\{x_{ij}(t)\}_{i,j=1\;,i\neq j}^n)\in \mathcal{B}$ being the current appraisal network at time $t$ and defined by the set  $\{x_{ij}(t)\}_{\{i,j\}\in E}$. To see this, we first notice that for any time step $t>0$ and any given appraisal network $G_X(t)$ described by the influence dynamics, it follows from the stochastic process of player selection in Definition~\ref{def-game}
that it is completely possible to determine all the possible outcomes for the appraisal network in the next time-step. Then, the Markov property is satisfied since the (distribution of) possible outcomes for the appraisal matrix at $t+1$ depends only on time index $t$. Moreover, the appraisal network $G_X(t)$ can have up to $|E|$ 
possible different outcomes. Then, the Markov chain $\mathcal{M}$ is well-posed, and it is also time-homogeneous from the time-invariance in the probability of player selection (see Definition~\ref{def-game}). 
%

If we assume there is a non-empty set of absorbing states  $\mathcal{W}_a$ in the Markov chain $\mathcal{M}$, then
to prove convergence to the absorbing states, we need to show that $\mathcal{B}\setminus\mathcal{W}_a$ is only composed of transient states~\citep{GG-DR:01}. Since $\mathcal{M}$ has finite states, the convergence will be achieved in finite time with probability one.


Consider $G_X(t)\in\mathcal{B}$ 
 and any selected edge $\{i,j\}\in E$. If $u_{ij}(t)\geq 0$, then $G_X(t+1)=G_X(t)$, and so $\Nop(G_X(t+1))=\Nop(G_X(t))$. Now, assume $u_{ij}(t)<0$, which implies $\Delta^u_{ij}>0$. Let~$\Delta^n_{ij}$ be the number of neutral triads in which agents $i$ and $j$ are part of. Notice that the cognitive dissonance function can be equivalently expressed as the number of unbalanced triads minus the number of neutral triads. 
Then, since the change of the interpersonal appraisal between $i$ and $j$ can only affect the nature (i.e., being balanced or unbalanced) of the $n-2$ triads that $i$ and $j$ are part of, it follows that:
\begin{enumerate}[label=(\roman*)]
\item If $x_{ij}(t)=+1$, then $x_{ij}(t+1)=-1$ and $\Delta^n_{ij}(t)=0$. Then, $\Delta_{ij}^b(t+1)=\Delta_{ij}^u(t)$, $\Delta_{ij}^u(t+1)=\Delta_{ij}^b(t)$ and $\Delta_{ij}^n(t+1)=\lambda_{ij}(t)$. Then, we have that
\begin{align*}
\Nop(G_X(t+1))-\Nop(G_X(t))&= (\Delta^u_{ij}(t+1)-\Delta^{n}_{ij}(t+1))-(\Delta^u_{ij}(t)-\Delta^{n}_{ij}(t))\\
&= (\Delta^b_{ij}(t)-\lambda_{ij}(t))-\Delta^u_{ij}(t) =u_{ij}(t) <0.
%
\end{align*}\label{uno1a}
%
%
%
%
%
\item If $x_{ij}(t)=-1$, then $x_{ij}(t+1)=+1$ and $\Delta^n_{ij}(t)=\lambda_{ij}(t)$. Then, $\Delta_{ij}^b(t+1)=\Delta_{ij}^u(t)$, $\Delta_{ij}^u(t+1)=\Delta_{ij}^b(t)$ and $\Delta_{ij}^n(t+1)=0$. Then, we have that
\begin{align*}
\Nop(G_X(t+1))-\Nop(G_X(t))&= (\Delta^u_{ij}(t+1)-\Delta^{n}_{ij}(t+1))-(\Delta^u_{ij}(t)-\Delta^{n}_{ij}(t))\\
&= \Delta^b_{ij}(t)-(\Delta^u_{ij}(t)-\lambda_{ij}(t)) =u_{ij}(t) <0.
\end{align*}\label{dos1a}
\end{enumerate}
From this analysis, we conclude that $\Nop(G_X(t+1))>\Nop(G_X(t))$. 

Now, from the best-response dynamics, if we have $u_{ij}(t)\geq 0$ for all $\{i,j\}\in E$, then $G_X(t+1)=G_X(t)$ with probability one. Therefore, we conclude that the set of absorbing states $\mathcal{W}_a$ is the set of Nash equilibrium networks (see the characterization on Lemma~\ref{lem1a}). 
 Then, we conclude that the sequence $(\Nop(G_X(t)))_t$ is non-increasing with probability one, and that, for any $G_X(t)\in\mathcal{B}\setminus\mathcal{W}_a$, there is always a positive probability of selecting an edge such that $\Nop(G_X(t+1))<\Nop(G_X(t))$. 
This let us conclude that $\mathcal{B}\setminus\mathcal{W}_a$ is composed of transient states. 

Finally, the second statement of the theorem 
follows directly from Theorem~\ref{lem_1co}. 
\end{proof}

\section{Additional discussion and results} 

\subsection{Connection to optimization and energy functions}
We observe from our convergence theorem that the influence dynamics attempts to solve the following combinatorial optimization problem
\begin{equation*}
\begin{aligned}
& \underset{\{x_{ij}\}\in\{-1,+1\}^{|E|}}{\text{minimize}}
& & \Nop(G_X).
\end{aligned}
%
\end{equation*}
%
We know a global minimum for this function corresponds to a network that has clustering balance.
The idea of proposing discrete optimization problems that a dynamic social balance model must solve
has been previously proposed in the physics community~\citep{SAM-SHS-JMK:09} only for the case of dynamic structural balance with the 
minimization of a potential energy function associated to generalized Ising models and complete social networks. 

\subsection{Numerical evidence}

In this section we show some numerical results about the convergence of the influence dynamics to clustering balance. We say that an initial condition at $t=0$ is \emph{generic} for the appraisal matrix $G_X(0)$ if every initial appraisal is independently sampled with probability $0<p<1$ of taking the value $+1$ and probability $1-p$ of taking the value $-1$. 

We analyze the evolution of appraisal networks for different network sizes $n\in\{3,\dots,25\}$. For each fixed $n$, we generate $10000$ generic initial conditions in which each entry of the initial appraisal network $G_X(0)$ takes the value $+1$ or $-1$ with equal probability $0.5$. The results are shown in Figure~\ref{plot_0}. As expected, for $n\leq 5$, all appraisal networks converged to a balanced network. Remarkably, we find that the success rate was no less than $99.98\%$ for $n\in\{15,\dots,25\}$, and no less than $99.72\%$ for $n\in\{6,\dots,14\}$. In Figure~\ref{plot_1} we show the empirical frequency of the final number of clusters for those networks that successfully converged to clustering balance for $n\in\{11,20\}$, and we remark that no convergence to efficient networks, i.e., to appraisal networks having structural balance, has been observed in these simulations.
%

Observing these and other numerical simulations for other generic initial conditions and greater values of $n$ (not shown here), we propose the following informal conjecture for $n>5$: 
under generic initial conditions, there is a very high probability of convergence to clustering balance and this probability goes to $1$ as $n$ increases. In other words, the basin of attraction of appraisal networks that have clustering balance is much larger than the ones for other Nash equilibrium networks.

\subsection{Connection to structural balance theory}
\label{con-sb}
We know that structural balance enforces the fourth Heider's rule, so we can modify our model in this paper so that this rule is enforced. This would mean changing our signed network formation game, more specifically, the utility function for a player $\{i,j\}$ simply becomes the difference between balanced and unbalanced triads.
We would also need to modify our influence dynamics by replacing any indicator function by a constant number $1$, and the cognitive dissonance function so that it simply counts unbalanced triads. As a result, 
the (modified) influence dynamics will seek to minimize this (modified) cognitive dissonance function, aiming to attain the minimum value $\Nop(G_X)=0$ for some appraisal network $G_X$.

We analyze the evolution of appraisal networks in the same settings as in the previous subsection. The results are shown in Figure~\ref{plot_0}. We also find that for $n\leq 5$, all appraisal networks converged to a network that has structural balance. For $n>5$, we found that the success rate was no less than $93.47\%$. It seems that the success rate stabilizes and oscillates around a value with the success rate for an odd number of agents being slightly better than for an even number of them. We cannot conclude if this observed behavior will be the case for much larger values of $n$, but we expect to.
%

As mentioned in the introduction of this paper, empirical evidence has suggested the strong presence of triads with all negative edges in online social networks, which favors the presence of the notion of clustering balance over structural balance. The contrast between the obtained numerical convergence results for clustering balance compared to ones for structural balance may be a theoretical predictor of the empirical observations.

%
%

%

 \begin{figure}[t]
   \centering
   \subfloat[Clustering balance]{\includegraphics[width=0.45\linewidth]{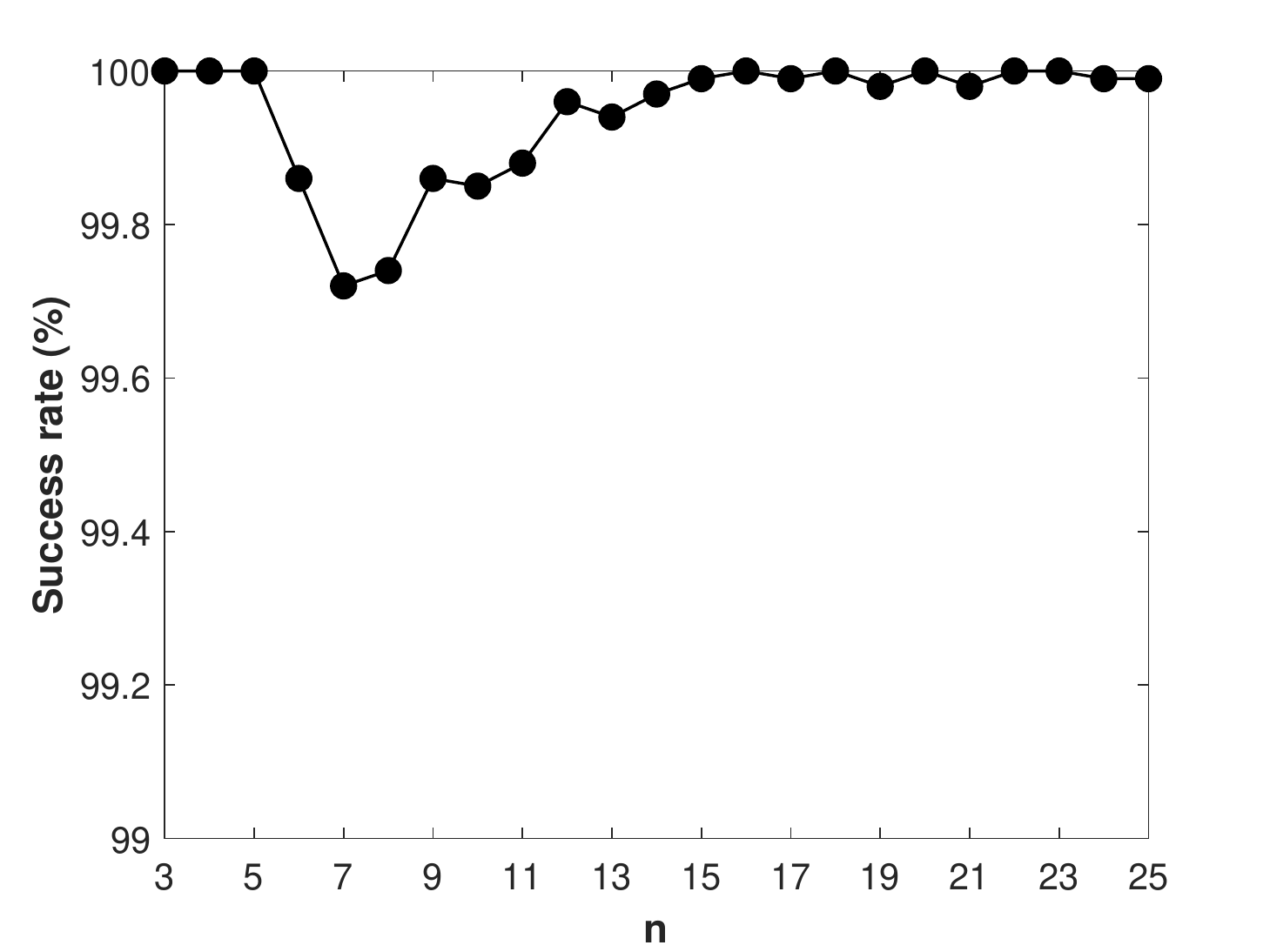}} 
  \subfloat[Structural balance]{\includegraphics[width=0.45\linewidth]{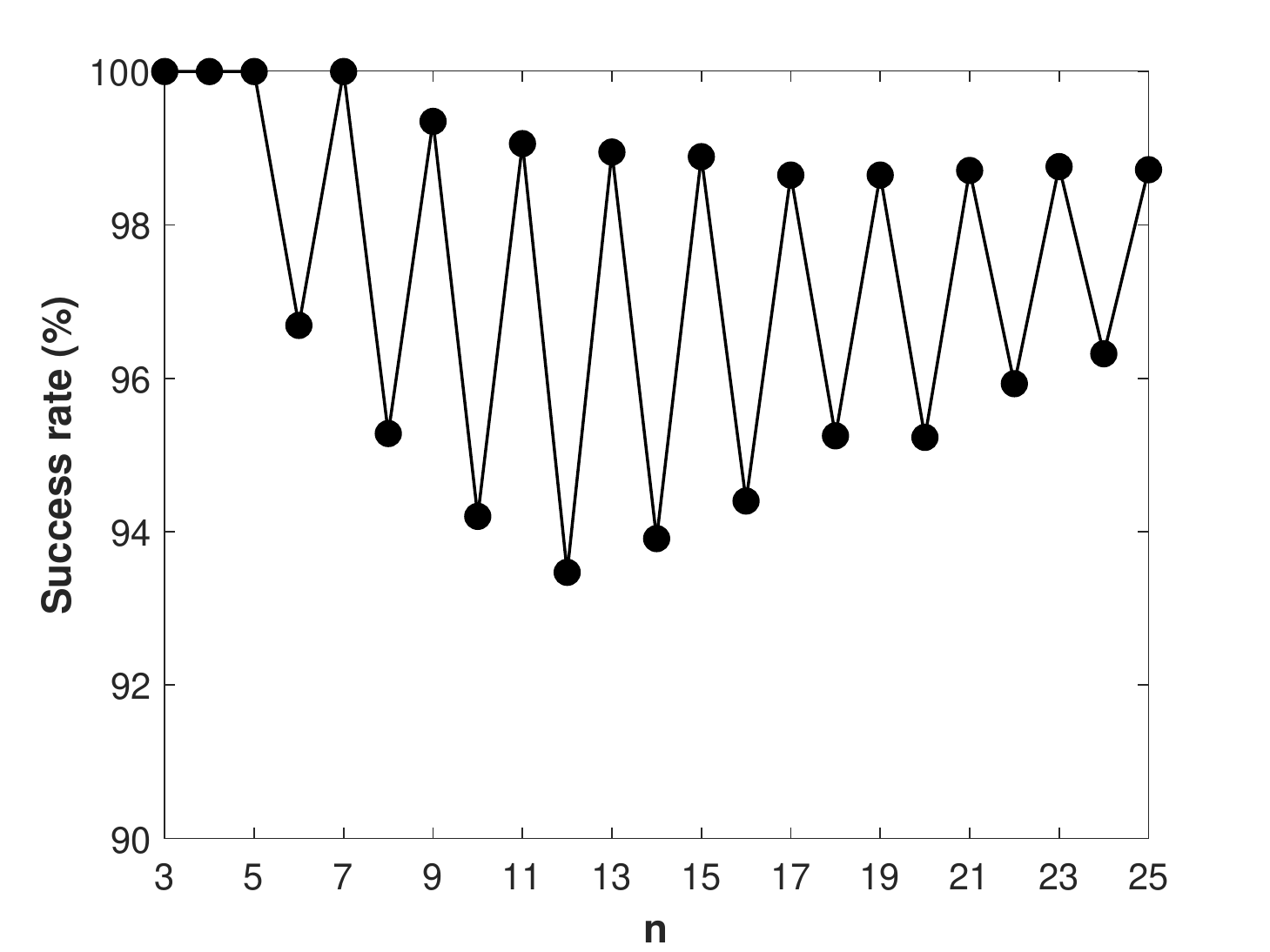}}
 \caption{
Success rate of convergence to (a) clustering balance (note range $99\%-100\%$) or (b) structural balance (note range $99\%-100\%$) for appraisal networks with different number of agents or nodes $n$. 
For each fixed size $n$, $10000$ simulations were performed under generic initial conditions. For clustering balance, the appraisal networks evolve according to the influence dynamics (Definition~\ref{eq_CODI}), and for structural balance, they evolve according to a modification on these dynamics as stated in subsection~\ref{con-sb}.
}
    \label{plot_0}
 \end{figure}

 \begin{figure}[t]
   \centering
   \subfloat[$n=11$]{\includegraphics[width=0.45\linewidth]{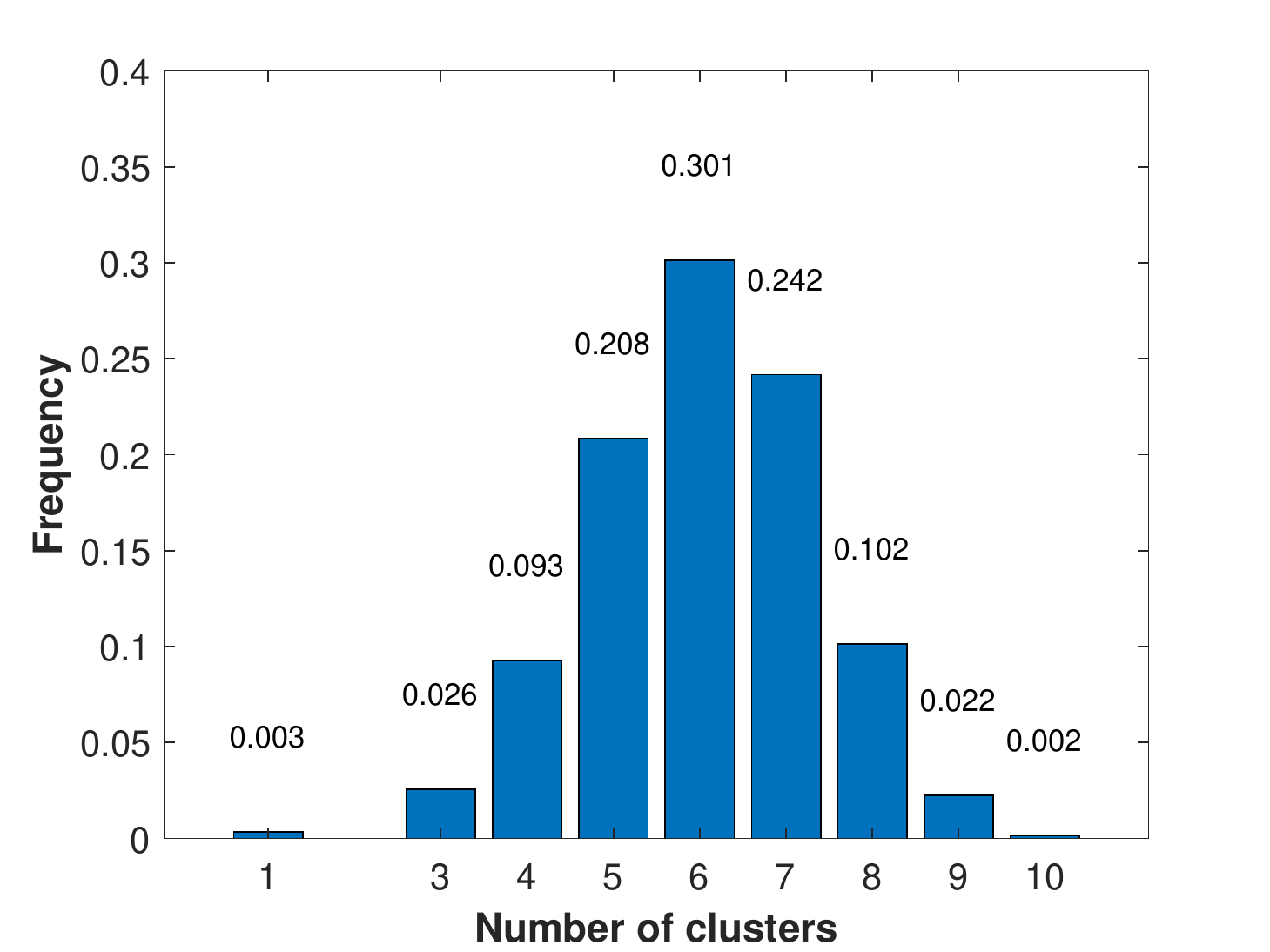}} 
  \subfloat[$n=20$]{\includegraphics[width=0.45\linewidth]{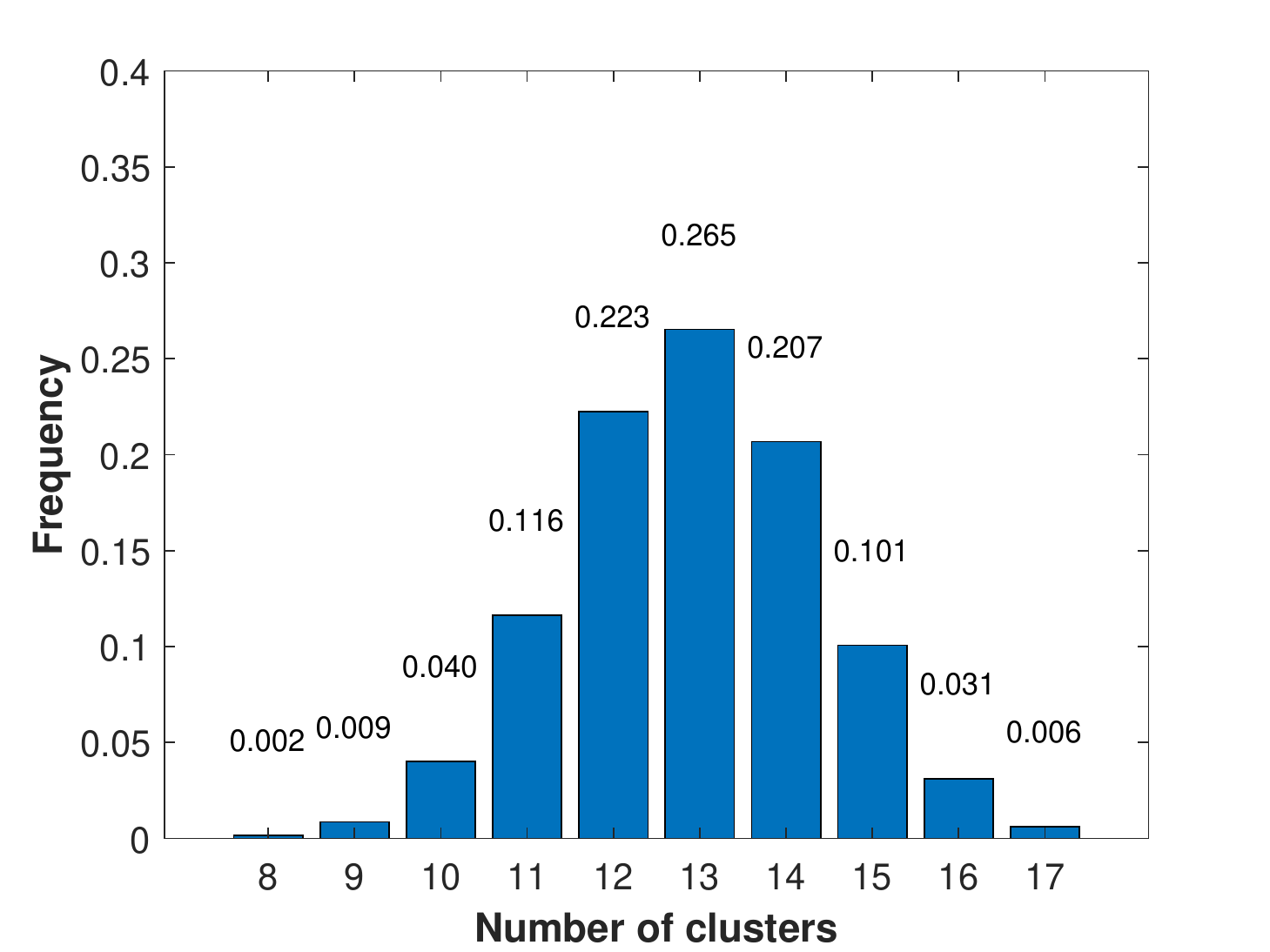}} 
 \caption{Empirical frequency of the final number of clusters for (complete) appraisal networks of sizes $11$ and $20$. For each plot, $10000$ simulations were performed with a randomly generated initial appraisal network. 
}
    \label{plot_1}
 \end{figure}

\section{Conclusion}

We propose, to the best of our knowledge, the first model of a signed
network formation game for the notion of clustering balance, whereby agents
can update only one interpersonal appraisal at a time. We have formally
shown how, in our proposed game, finding a Nash equilibrium can provide a model for
dynamic clustering balance. Moreover, our model has a psycho-sociological
interpretation in which the best-response policy results in the eradication
of cognitive dissonances among individuals in a social network. However,
broader interpretations can be given using the same underlying model; for
example, the interpretation of how countries or communities change their
positive or negative diplomatic relationships in order to avoid or create
conflict respectively. This makes our work relevant to the fields of
economics of conflict and political science. Finally, our model's
relationship to potential energy functions and combinatorial optimization
may make this work relevant to the physics and mathematical communities.
As future work, motivated by the structure of real-world networks, we plan to study signed network formation games 
that can consider the evolution of signed graphs that are not complete.
We also consider important to further study the structural balance dynamic model presented in subsection~\ref{con-sb}, and also study the evolution of directed signed graphs 
that can
allow the generation of dynamic models for other notions of social
balance.


\bibliographystyle{abbrvnat}
\bibliography{alias,Main,FB}

\end{document}